\documentclass[12pt]{amsart}

\usepackage{amscd,amsmath}
\usepackage{amssymb,amsmath,latexsym}

\usepackage[all]{xy}
\usepackage[pdftex]{graphicx}

\newtheorem{theorem}{Theorem}[section]

\newtheorem{proposition}[theorem]{Proposition}

\begin{document}

\title{Complexity of Correspondence Homomorphisms}
        
\author[T. Feder]{Tom\'{a}s Feder}
\address{268 Waverley Street \\
               Palo Alto, CA 94301, USA}
\email{tomas@theory.stanford.edu}

\author[P. Hell]{Pavol Hell}
\address{School of Computing Science \\
 	      Simon Fraser University \\
              Burnaby, B.C., Canada V5A 1S6}
\email{pavol@sfu.ca}

\date{}

\maketitle

\begin{abstract}

Correspondence homomorphisms are both a generalization of standard homomorphisms 
and a generalization of correspondence colourings. For a fixed target graph $H$, the 
problem is to decide whether an input graph $G$, with each edge labeled by a pair of 
permutations of $V(H)$, admits a homomorphism to $H$ `corresponding' to the labels, 
in a sense explained below.

We classify the complexity of this problem as a function of the fixed graph $H$. It turns out 
that there is dichotomy -- each of the problems is polynomial-time solvable or NP-complete.
While most graphs $H$ yield NP-complete problems, there are interesting cases of graphs $H$ 
for which the problem is solved by Gaussian elimination.

We also classify the complexity of the analogous correspondence {\em list homomorphism} problems, 
and also the complexity of a {\em bipartite version} of both problems. We emphasize the proofs for the 
case when $H$ is reflexive, but, for the record, we include a rough sketch of the remaining
proofs in an Appendix.
\end{abstract}

\section{Introduction}

Let $H$ be a fixed graph, and let $P(H)$ denote the set of all permutations of $V(H)$.
Let $G$ be a graph with an edge-labeling $\ell : E(G) \to P(H) \times P(H)$. Each edge
$xy$ of $G$ is assumed to have a label $(\pi,\rho)$ where $\pi$ is viewed as associated
with $x$ and $\rho$ associated with $y$. An $\ell$-{\em correspondence homomorphism} 
of $G$ to $H$ of $G$) is a mapping $f : V(G) \to V(H)$ such that $xy \in E(G)$ with 
$\ell(xy)=(\pi,\rho)$ implies $\pi(f(x))\rho(f(y)) \in E(H)$. When $H$ is the irreflexive
(i.e., loopless) complete graph on $k$ vertices, correspondence homomorphisms to
$H$ have been called {\em correspondence $k$-colourings}, and applied to answer a 
question of Borodin on choosability of planar graphs with certain excluded cycles \cite{dp}. 
These colorings have also been called DP-colourings, in honour of the authors of \cite{dp},
and have proved quite interesting from a graph theoretic point of view. cf. \cite{xuding,wang} 
and the references therein. To emphasize this connection, we sometimes call correspondence 
homomorphisms to $H$ also {\em correspondence $H$-homomorphisms}, or {\em correspondence
$H$-colourings}.

The {\em correspondence $H$-homomorphism problem} takes as input a graph 
$G$ with labeling $\ell$ and asks whether or not an $\ell$-correspondence homomorphism to 
$G$ exists. In the {\em correspondence $H$-list-homomorphism problem}, $G$ is also equipped 
with {\em lists} $L(x), x \in V(G)$, each a subset of $V(H)$, and the $\ell$-correspondence 
homomorphism $f$ also has to satisfy $f(x) \in L(x), x \in V(G)$. Clearly, these problems 
are generalizations of the well-known $H$-homomorphism (i.e., $H$-colouring) and list 
$H$-homomorphism (list $H$-colouring) problems respectively. Those problems are 
obtained as special cases when $\ell$ chooses two identity permutations on each edge 
of $G$. They have been studied in \cite{fh,fhh,fhh2,hn}, cf. \cite{hombook}.

There is another way to think of the correspondence homomorphism problem, one that is often 
helpful in the proofs and illustrations. Let $G, \ell$ be an instance of the correspondence 
$H$-homomorphism problem. Construct a new graph $G^*$ by replacing each vertex $x$ of 
$G$ with its own separate copy $V_x$ of the set $V(H)$, with the following edges. If $xy$ is 
an edge of $G$ labelled by $\ell(xy) = (\pi,\rho)$, we join $V_x$ and $V_y$ with the edges 
from $\pi(x)(u) \in V_x$ to $\rho(y)(v) \in V_y$ for all edges $uv$ of $H$. (Recall that  $\pi(x)$ 
is a permutation of $V(H)$, and we view it as also a permutation of $V_x$, and similarly for  
$\rho(y)$.) Then an $\ell$-correspondence homomorphism corresponds precisely to a transversal 
of the sets $V_x, x \in V(G)$, (a choice of exactly one vertex from each $V_x$) which induces in 
$G^*$ an isomorphic copy of $G$. Consider the edges between two adjacent sets $V_x, V_y,$
$xy \in E(G)$. Recall the bipartite graph $H'$ {\em associated with} $H$, in which each 
$v \in V(H)$ yields two vertices $v_1, v_2$ in $H'$ and each edge $uv \in E(H)$ yields two 
edges $u_1v_2, u_2v_1$ in $E(H')$. The edges in $G^*$ between adjacent sets $V_x, V_y$ 
form an isomorphic copy of $H'$, with the part $V_x$ permuted according to $\pi$, and the
part in $V_y$ according to $\rho$. The label $\ell(xy)$ specifies the way $H'$ is layed out 
between the sets $V_x$ and $V_y$.

In this note we focus on the reflexive case, i.e., we assume that $H$ has a loop at every
vertex. While the standard $H$-homomorphism problem is trivial for reflexive graphs, the
correspondence $H$-homomorphism problem turns out to be more interesting. Moreover, it makes 
sense to consider inputs $G$ that may have loops and parallel edges, as the permutation 
constraints on these may introduce significant restrictions. We will use this freedom in Section 
4, to simplify NP-completeness proofs. We will also show, in Section 2, that the complexity
of the problem for graphs with loops and parallel edges allowed or forbidden are the same.

A {\em reflexive clique} is a complete graph with all loops, a {\em reflexive co-clique} is a
set of disconnected loops with no other edges. A disjoint union of cliques $K_p$ and $K_q$
will be denoted by $K_p \cup K_q$ or $2K_p$ if $p=q$, and similarly for more than two
cliques. We use the same symbols for reflexive cliques and irreflexive cliques, and the
right context will always be specified or clear from the context.

Our main result is the following dichotomy classification of both the correspondence 
homomorphism problem and the correspondence list homomorphism problem.

\begin{theorem}\label{main}
Suppose $H$ is a reflexive graph.

If $H$ is a reflexive clique, a reflexive co-clique, or a reflexive $2K_2$, then the correspondence
$H$-homomorphism problem is polynomial-time solvable. In all other cases, the correspondence
$H$-homomorphism problem is NP-complete.

If $H$ is a reflexive clique or a reflexive co-clique, then the correspondence $H$-list-homomorphism 
problem is polynomial-time solvable. In all other cases the correspondence $H$-list-homomorphism 
problem is NP-complete.
\end{theorem}

In the last section we state the analogous result for general graphs (with possible loops) 
and for bipartite graphs,
and in the Appendix we provide a rough sketch of the proofs.

\section{Loops and Parallel Edges}

Suppose $H$ is a fixed reflexive graph, and $G$ is a graph with loops and parallel edges 
allowed, with an edge-labeling $\ell$. Thus there are labels on loops, and parallel edges 
may have different labels. We will construct a modified simple graph $G'$ (without loops 
and parallel edges), and a modified edge-labeling $\ell'$ on $G'$ such that $G$ has an
$\ell$-correspondence homomorphism to $H$ if and only if $G'$ has an $\ell'$-correspondence
homomorphism to $H$. The changes from $G, \ell$ to $G', \ell'$ proceed one loop and
one pair of parallel edges at a time.

Suppose $G$ has a loop $xx$ with label $\ell(xx) = (\pi(x),\rho(x))$. Replace $x$ by a
clique with vertices $x_0, x_1, x_2, \dots, x_n$ where $n=|V(H)|$. Each edge $x_ix_j$
will have the same label $\ell(x_ix_j) = (\pi(x),\rho(x))$. Each vertex $x_i$ will have the
same adjacencies, with the same labels, as $x$ did. Call the resulting graph $G_1$ and
the resulting labeling $\ell_1$. Then we claim that $G$ has an $\ell$-correspondence 
homomorphism to $H$ if and only if $G_1$ has an $\ell_1$-correspondence homomorphism 
to $H$. The one direction is obvious, if $f$ is an $\ell$-correspondence homomorphism of $G$
to $H$, then the same mapping, extended to all copies $x_i$ of $x$ is an $\ell_1$-correspondence 
homomorphism of $G_1$ to $H$. Conversely, suppose that $f$ is an $\ell_1$-correspondence 
homomorphism of $G_1$ to $H$. This gives $n+1$ values $f(x_i)$ among the $n$ possible
images $V(H)$. Thus the majority value must appear on at least two distinct vertices $v_i, v_j$,
and therefore $\pi(v_i)\rho(v_j)$ is an edge of $H$. Thus the mapping $F$ which assigns to
$x$ the majority value amongst $f(x_i)$ and equals to $f$ on all other vertices is an $\ell$-
correspondence homomorphism of $G$ to $H$.

Parallel edges are removed by a similar trick using expanders instead of cliques. To be
specific, assume that $(xy), (xy)'$ are two different parallel edges of $G$ joining the same 
vertices $x$ and $y$, with labels $\ell((xy)) = (\pi(x),\rho(y))$ and $\ell'((xy)') = (\pi'(x),\rho'(y))$.
Replace $x$ and $y$ by a large number $N$ of vertices $x_i, y_j$, each having the same
adjacencies and labels to other vertices as $x, y$ respectively, and all edges $x_iy_j$ for
all $i$ and $j$. Some of the edges $x_iy_j$ are labelled by $\ell((xy))$ and the others by 
$\ell'((xy)')$. The set of edges labelled by $\ell((xy))$ defines a bipartite graph $B$ and those
labelled by $\ell'((xy)')$ form its bipartite complement $\overline{B}$. With the right choice of
$B$ we will be able to conclude that between any large sets of $x_i$'s and $y_j$'s there is
at least one edge of $B$ and at least one edge of $\overline{B}$. This allows the above idea
of using majority values to work, namely if $f$ is a correspondence homomorphism on the replaced
graph, we may define $F(x)$ to be the majority value of $f(x_i)$ and similarly for $F(y)$. Then
the two values appear at both an edge labelled by $\ell((xy))$ and an edge labelled by $\ell'((xy)')$,
and so $F$ is a correspondence homomorphism on the original graph $G$.

There are many known proofs that such expanders $B$ do exist \cite{something}. It is also not hard 
to see directly. For instance, if $|V(H)|=n$, let $N > n \log_2 n$ and take for $B$ a random bipartite 
graph on $N$ versus $N$ vertices. The number of ways of choosing sets of size 
$$t = \lceil N/n \rceil > 3 \log_2 n$$
for both sides is $A={{N} \choose{t}}^2\leq {(\frac{eN}{t})}^{2t}\leq 2^{2tlog_2 N}$. The probability that 
the sets will not be joined by a random choice of edges is bounded by $\frac{1}{2^{t^2-1}} < 1/A$, 
giving a positive probability to the existence of a suitable bipartite graph $B$.

\section{Polynomial Cases for Reflexive $H$}

If $H$ is a reflexive clique, the correspondence $H$-homomorphism problem can be trivially solved.
Each vertex $x$ of the input graph $G$ can be assigned any image, and regardless of the 
permutation labels $\ell$, the mapping is a perm-homomorphism. In fact, this observation 
also solves the correspondence $H$-list-homomorphism problem.

If $H$ is a reflexive co-clique, the correspondence $H$-homomorphism problem also has an easy
solution, since every choice of an image for a vertex $x$ of $G$ implies a unique image of any 
adjacent vertex $y$. Thus for each component of $G$ we may try all images of a particular 
vertex $x$, and the component can be mapped to $H$ if and only if one of these images 
produces an $\ell$-consistent homomorphism. This also works for the correspondence list
$H$-homomorphism problem problem.

The most interesting case occurs when $H = 2K_2$. We name the vertices of $H$
by binary 2-strings, $00, 01, 10, 11$. In addition to the loops $00-00, 01-01, 10-10,
11-11$, the two edges of $H$ are, say, $00-01$ and $10-11$. Note that because of
symmetry, there are only three different permutations of $H$, this one, with edges 
$00-01$ and $10-11$, and two more, with edges $00-10$ and $01-11$, or with edges
$00-11$ and $01-10$. Consider the alternate view via the auxiliary graph $G^*$ 
discussed earlier. For each edge $xy$ of an input graph $G$, the edges between 
$V_x$ and $V_y$ depend only on which of the above possibilities apply to each of 
$V_x$ and $V_y$, because of the nature of the adjacencies, in the form of two copies
of $K_{2,2}$. Moreover, it does not matter which vertex on each side of the $K_{2,2}$
is chosen. (Here we use $K_{2,2}$ to denote the {\em irreflexive} complete bipartite
graph with two vertices on each side; see Figure 1.) Let us associate with each vertex 
$x$ of $G$, two $\{0,1\}$ variables $x_a, x_b$, to be understood as describing the two 
coordinates of the name of the chosen vertex for $V_x$. Each of the above partitions 
of $V(H)$ into two edges can now be described by linear equations modulo two:

\begin{itemize}
\item
$00-01$ and $10-11$ correspond to the equations $x_a = 0$ and $x_a = 1$
\item
$00-10$ and $01-11$ correspond to the equations $x_b = 0$ and $x_b = 1$
\item
$00-11$ and $01-10$ correspond to the equations $x_a + x_b = 0$ and $x_a + x_b = 1$
\end{itemize}

In Figure 1, we have the partition on $V_x$ corresponding to the first bullet,
and the partition of $V_y$ corresponding to the second bullet, while the edges of
$2K_{2,2}$ joint the first part of the partition on $V_x$ with the second part of the
partition on $V_y$ and vice versa. To satisfy the constraints of correspondence,
we have to make sure that if $x$ selects a vertex in the part with $x_a = 0$, then 
$y$ selects a vertex in the part with $y_a + y_b = 1$, and if $x$ selects in $x_a = 1$
then $y$ selects in $y_a + y_b = 0$. These constraints can be expressed by the
linear equation $x_a + y_a + y_b = 1$. In all the other cases it is also easy to check 
that there is a linear equation modulo two, which describes the constraints. Thus 
we have reduced the problem of existence of a correspondence $2K_2$-homomorphism 
to a system of linear equations modulo two, which can be solved in polynomial-time 
by Gaussian elimination.

\begin{figure}[hhhh]
\includegraphics[height=6cm]{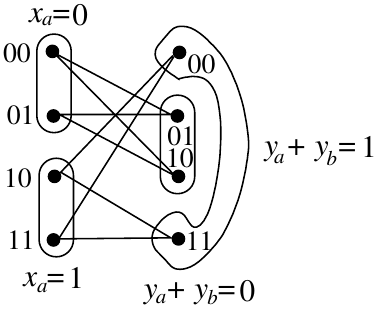}
\caption{The edges are described by the equation $x_a + y_a + y_b = 1$ modulo two}
\end{figure}

It turns out, see the next section, that the list version of the correspondence $2K_2$-homomorphism
problem is NP-complete.

\section{NP-complete Cases for Reflexive $H$}

We begin with the following simple example. Suppose $H$ is $(K_1 \cup K_2)$, the disjoint
union of a reflexive  cliques $K_1$ and $K_2$. Specifically, let $(K_1 \cup K_2)$ have a loop 
on $a$, and two adjacent loops on $b$ and $c$.

\begin{proposition}
The correspondence $(K_1 \cup K_2)$-homomorphism problem is NP-complete.
\end{proposition}

\begin{proof}
We give a reduction from 1-IN-3-SAT (without negated variables). Consider an instance,
with variables $x_1, x_2, \dots, x_n$ and triples (of variables) $T_1, T_2, \dots, T_m$. 
Construct the corresponding instance of the correspondence $H$-homomorphism problem 
as follows. The instance $G$ will contain vertices $v_1, \dots, v_n$, as well as $T_1, \dots, T_m$,
and all edges are of the form $v_iT_j$ with $v_i$ appearing in the triple $T_j$. The variables
in the triples are arbitrarily ordered, each $T_j$ having a first, second, and third variable.
The edge between $T_j$ and its first variable $v_p$ is labeled so that the special vertex $a$ 
in $V_{v_p}$ is adjacent to $a$ in $V(T_j)$, the edge between $T_j$ and its second variable 
$v_q$ is labeled so that the special vertex $a$ in $V_{v_q}$ is adjacent to $b$ in $V(T_j)$,
and similarly for the third variable and $c$ in $V(T_j)$. (See Figure 2.) Because of the
adjacencies, if any one vertex $x$ is chosen in $V(T_j)$, there is exactly one of its variables
$v_p, v_q, v_r$ that has its special vertex $a$ adjacent to $x$.

\begin{figure}[hhhh]
\includegraphics[height=9cm]{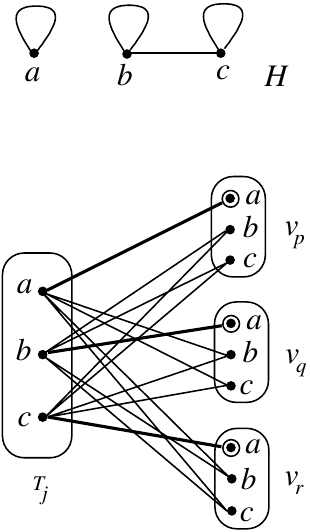}
\caption{The triple $T_j$ with variables $v_p, v_q, v_r$, considered in that order}
\end{figure}

We claim that this instance has a correspondence homomorphism if and only if the original 
instance of 1-IN-3-SAT was satisfiable. Indeed, any satisfying truth assignment sets as true a set 
of $v_i$'s such that exactly one appears in each $T_j$; hence we can set the value of the vertex 
$v_i$ to be $a$ whenever the variable $v_i$ was true in the truth assignment (and, say, $b$ 
otherwise). The value of the vertex $T_j$ will be $a$ if its first variable was true, $b$ if its second
variable was true, and $c$ if its third variable was true. It is easy to see that this defines a 
correspondence $(K_1 \cup K_2)$-homomorphism. For the converse, observe that a
correspondence homomorphism selects one vertex from each $V(T_j)$, which forces
exactly one of its variable to select the value $a$, thus defining a satisfying truth assignment.
\end{proof}

We note that this result implies that the list version of the correspondence $2K_2$-homomorphism 
problem is also NP-complete, since the correspondence $K_1 + K_2$-homomorphism problem reduces
to the correspondence $2K_2$-list-homomorphism problem. Indeed, lists may be used to restrict the 
input vertices never to use one of the four vertices of $2K_2$.

A similar reduction from 1-IN-t-SAT shows the following fact.

\begin{proposition}
The correspondence $(K_1 \cup K_t)$-homomorphism problem is NP-complete, for all $t \geq 2$.
\end{proposition}

The method we use most often is described in the following proposition.

\begin{proposition} \label{last}
The correspondence $(K_p \cup K_q \cup \dots \cup K_z)$-homomorphism problem reduces
to the correspondence $(K_{p+1} \cup K_{q+1} \cup \dots \cup K_z)$-homomorphism problem.
\end{proposition}

\begin{proof}
Given an instance $G$ of the $(K_p \cup K_q \cup \dots \cup K_z)$-homomorphism 
problem, we create a new graph $G'$ by adding at each vertex $v$ of $G$ a loop 
$e_v$ (even if $v$ may already have loops), labelled by two permutations $(\pi(e_v),\rho(e_v))$ 
that forbid one vertex of $K_p$ and one vertex of $K_q$. Specifically, suppose we want to forbid 
a vertex $a$ of $K_p$ and a vertex $b$ of $K_q$. Choose $\pi(e_v)$ to be the identity and $\rho(e_v)$ 
to be the involution exchanging $a$ and $b$. Now we claim that $G'$ has a correspondence 
$(K_p \cup K_q \cup \dots \cup K_z)$-homomorphism if and only if the original graph $G$ 
has a correspondence $(K_{p+1} \cup K_{q+1} \cup \dots \cup K_z)$-homomorphism. Indeed, if $f$ 
is a correspondence $(K_p \cup K_q \cup \dots \cup K_z)$-homomorphism for the target graph has
neither $a$ nor $b$, and so, the added loops create no problem, i.e.,  $f$ remains a correspondence 
$(K_{p+1} \cup K_{q+1} \cup \dots \cup K_z)$-homomorphism. On the other hand, if $f$ is a
$(K_{p+1} \cup K_{q+1} \cup \dots \cup K_z)$-homomorphism of $G'$, the label on $e_v$
ensures that $v$ cannot map to $a$ or $b$ as neither $\pi(e_v)(a) = a$ is adjacent to 
$\rho(e_v)(a) = b$ nor $\pi(e_v)(b) = b$ is adjacent to $\rho(e_v)(b) = a$. Thus $f$ is also
a $(K_p \cup K_q \cup \dots \cup K_z)$-homomorphism.
\end{proof}

Proposition \ref{last} allows us to prove the NP-completeness of the correspondence
$(K_{p+1} \cup K_{q+1} \cup \dots \cup K_z)$-homomorphism problem from 
the NP-completeness of the correspondence $(K_p \cup K_q \cup \dots \cup K_z)$-homomorphism 
problem. Let us call the operation of removing a vertex from each of two distinct cliques  
a {\em pair-deletion}. Thus any $H$ that is a union of reflexive cliques that can be 
transformed to $K_1 \cup K_2$ by a sequence of pair-deletions yields an NP-complete
correspondence $H$-homomorphism problem.

We now consider two additional important special cases, starting with $H = 2K_3$ 
consisting of two reflexive triangles. Note that any pair-deletion of $2K_3$ produces
a $2K_2$ which has a polynomial-time solvable problem.

\begin{proposition}
The correspondence $2K_3$-homomorphism problem is NP-complete.
\end{proposition}

\begin{proof}
Assume $H = 2K_3$ has triangles $abc$ and $a'b'c'$. We add to each vertex $v$ of $G$ 
two separate loops $e_v, e'_v$ with labels designed to make it impossible to map $v$ to 
$a$ or $a'$ or $b'$. This reduces the NP-complete correspondence $(K_1 \cup K_2)$-homomorphism 
problem to the correspondence $2K_3$-homomorphism problem. The first loop $e_v$ will have the label 
$(\pi(e_v),\rho(e_v)$ where $\pi(e_v)$ is the identity and $\rho(e_v)$ is the involution exchanging
$a$ and $a'$.  The second loop $e'_v$ will have the label $(\pi(e'_v),\rho(e'_v))$ where $\pi(e'_v)$ 
is the identity and $\rho(e'_v)$ is the involution exchanging $a$ and $b'$. Now we claim that the 
resulting graph $G'$ has a correspondence $(K_1 \cup K_2)$-homomorphism for the target graph on 
$b, c, c'$ (isomorphic to $K_1 \cup K_2$) if and only if the original graph $G$ has a correspondence
$2K_3$-homomorphism. If $f$ is a correspondence $(K_1 \cup K_2)$-homomorphism for the target 
graph on $b, c, c'$, neither of loops creates a problem, and $f$ remains correspondence to $H$. The
converse follows from the fact that the first loop forbids $a$ and $a'$ for any vertex $v$, since
neither $\pi(e_v)(a) = a$ is adjacent to $\rho(e_v)(a) = a'$ nor $\pi(e_v)(a') = a'$ is adjacent to 
$\rho(e_v)(a') = a$, and similarly the second loop $e'v$ forbids $a$ and $b'$.
\end{proof}

The second example is the union of two isolated loops $a, b$ and two adjacent loops $c, d$,
i.e., the graph $H = K_1 \cup K_1 \cup K_2$. It also does not admit a pair-deletion that
results in $K_1 \cup K_2$.

\begin{proposition}
The correspondence $(K_1 \cup K_1 \cup K_2)$-homomorphism problem is NP-complete.
\end{proposition}

\begin{proof}
We reduce the NP-complete correspondence $(K_1 \cup K_3)$-homomorphism problem to the 
correspondence $(K_1 \cup K_1 \cup K_2)$-homomorphism problem. Suppose $G, \ell$ is an instance 
of the correspondence $(K_1 \cup K_3)$-homomorphism problem with the target graph consisting of
the following vertices inducing $K_1 \cup K_3$: an isolated loop at $a$ and a reflexive triangle $bcd$.
We form a new graph $G'$ and labeling $\ell'$ by replacing each edge $xy$ of $G$ by a path $xx', x'y', y'y$
so that if $\ell(xy) = (\pi,\rho)$ then $\ell'(xx') = (\pi,\pi')$, $\ell'(x'y') = (\pi,\pi')$, and $\ell'(y'y) = (\pi,\rho')$
where $\pi'$ is obtained by composing $\pi$ with the involution exchanging $b$ and $c$, and $\rho'$ by composing
$\rho$ with the same involution of $b$ and $c$. (See Figure 3). We observe that the path $xx', x'y', y'y$
admits a corresponding homomorphism to $K_1 \cup K_1 \cup K_2$ taking $x$ and $y$ to any of the
pairs $aa, bb, cc, dd, bc, cb, bd, db, cd, dc$, but never $a$ with another vertex. It is easy to see that
this implies that $G$ admits an $\ell$-correspondence homomorphism to $K_1 \cup K_3$ on $a, b, c, d$ 
if and only if $G'$ admits an $\ell'$-correspondence homomorphism to $K_1 \cup K_1 \cup K_2$ with edge 
$cd$ (and all loops). (This is a correspondence version of the 'indicator construction' from \cite{hn,hombook}, 
where the method is discussed in more detail.)

\begin{figure}[hhhh]
\includegraphics[height=6cm]{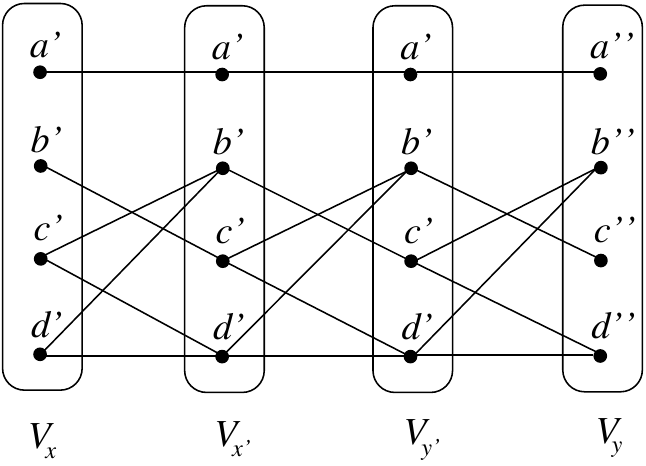}
\caption{A replacement for the edge $xy$ labeled by $\pi=a'b'c'd'$ and $\rho=a''b''c''d''$}
\end{figure}

It is now easy to check that any union of reflexive cliques other than a single clique,  a union of
disjoint $K_1$'s (i.e., a co-clique), or $2K_2$, can be transformed by pair-deletions to one of
the NP-complete cases $K_1 \cup K_2, K_1 \cup K_1 \cup K_2, 2K_3$, and hence they are 
all NP-complete.
\end{proof}

Now consider a reflexive target graph $H$ that is not a union of cliques. The {\em square} $H^2$ of a
graph $H$ has the same vertex-set $V(H^2)=V(H)$, and two vertices are adjacent in $H^2$ if and only
if they have distance at most two in $H$. The following observation is useful.

\begin{proposition}\label{indi}
The correspondence $H^2$-homomorphism problem reduces to the correspondence $H$-homomorphism 
problem.
\end{proposition}

\begin{proof}
Let $G'$ be obtained from $G$ by subdividing each edge $xy$ to be two edges $xz, zy$. If $\ell(xy) = (\pi,\rho)$
then the label of $xz$ is $(\pi,1)$, and the label of $zy$ is $(1,\rho)$. Now it can be seen that $G$ has a 
correspondence $H^2$-homomorphism if and only if $G'$ has a correspondence $H$-homomorphism. 
(We again apply the logic of the indicator construction \cite{hombook}, since the path $xz, zy$ can map 
$x$ and $y$ to any edge of $H^2$.)
\end{proof}

It follows that if $H$ is not connected, we can deduce the NP-completeness of the correspondence
$H$-homomorphism problem from the above results on the union of reflexive cliques (with the
only exceptions of $2K_2$ and $tK_1$). For connected $H$, we can apply Proposition \ref{indi}
to a sufficiently high power of $H$ that has diameter two.

\begin{proposition}\label{two}
If $H$ has diameter two but is not the reflexive path of length two, then he correspondence 
$H$-homomorphism problem is also NP-complete.
\end{proposition}

\begin{proof}
There must, in $H$, be a path $ab, bc$ where $a$ and $c$ are not adjacent. Suppose first
that $H - a$ is not a clique. Then the correspondence $(H-a)$-homomorphism problem can
be assumed NP-complete by induction (on $|V(H)|$), and it reduces to the correspondence 
$H$-homomorphism problem as follows. Suppose $G$ is an instance of the 
$(H-a)$-homomorphism problem, and form $G'$ by adding at each vertex of $G$ a loop 
labelled $(\pi,\rho)$ where $\pi$ is the identity and $\rho$ is the cyclic permutation $(a,b,c)$.
The effect of these loops is to prevent any vertex from mapping to $a$, as $aa$ is not equal
to $\pi(u)\rho(v)$ for any edge $uv \in E(H)$ (but $bb$ and $cc$ are, and even though 
$cb$ is not an edge $bc$ is an edge, which is sufficient for a loop). If $H - c$ is not a
clique we proceed analogously. Otherwise there exists a vertex $d$ adjacent to $a, b, c$.
In this case we can add a loop to each vertex of $G$ that effectively deletes $b$ and we 
can again apply induction on $|V(H)|$. These loops will have labels $(\pi,\rho)$ where 
$\pi$ is the involution of $a$ and $b$ and $\rho$ is the involution of $b$ and $c$. (See 
Figure 4.) The argument is similar, noting that only $bb$ is missing, because $ca$ is
also equal to $ac$ for a loop.
\end{proof}

\begin{figure}[hhhh]
\includegraphics[height=5cm]{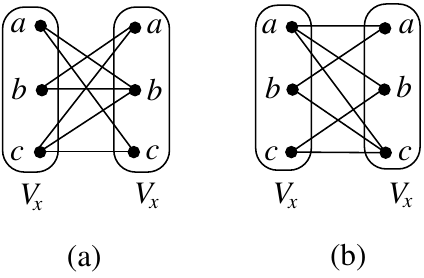}
\caption{(a) Labelling of a loop to remove $a$; (b) labelling of a loop to remove $b$}
\end{figure}

\begin{proposition}\label{one}
If $H$ is the reflexive path of length two, then he correspondence $H$-homomorphism problem is 
NP-complete.
\end{proposition}

\begin{proof}
Suppose $H$ is the reflexive path $ab, bc$. We reduce from $3$-colourability. Thus suppose
$G$ is an instance of $3$-colourability, and form $G'$ by replacing each edge $xy$ of $G$
by two parallel edges $(xy)_1, (xy)_2$ with permutations $(\pi_1,\rho_1)$ on $(xy)_1$ and
$(\pi_2,\rho_2)$ on $(xy)_2$. Both $\pi_1$ and $\rho_1$ are identity permutations, both 
$\pi_2$ and $\rho_2$ are the involutions exchanging $a$ and $b$. They are illustrated in 
Figure 5. The effect of $(xy)_1$ is that if $x$ maps to $a$, then $y$ maps to either $b$ or 
$c$ but not $a$, and if $x$ maps to $c$ then $y$ maps to either $a$ or $b$ but not $c$. 
The effect of the second edge $(xy)_1$ does not preclude any of these possible images, 
but also restricts $b$ to map to $a$ or $c$ but not $b$. Thus a correspondence homomorphism 
of $G'$ is a $3$-colouring of $G$ with the colours $a, b, c$. The converse is also easy to see.
\end{proof}

\begin{figure}[hhhh]
\includegraphics[height=5cm]{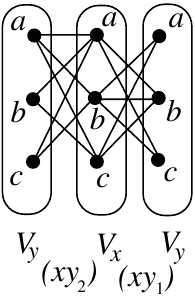}
\caption{Labelling of parallel edges $(xy)_1$ and $(xy)_2$}
\end{figure}

This completes the proof of Theorem \ref{main}.

\section{The General Results}

We have similar results for general graphs $H$, where some vertices may have loops and
others not. Note that in this case isolated loopless vertices can be removed from $H$ without 
affecting the complexity of the correspondence $H$-homomorphism problem or list 
$H$-homomorphism problem.

\begin{theorem}\label{mixed}
Let $H$ be a graph with possible loops. Suppose moreover that if $H$ has both a vertex
with a loop and a vertex without a loop, then it has no isolated loopless vertices.

\vspace{2mm}

The following cases of the correspondence $H$-homomorphism problem are polynomial-time 
solvable.

\begin{enumerate}
\item
$H$ is a reflexive clique
\item
$H$ is a reflexive co-clique
\item
$H$ is a reflexive $2K_2$
\item
$H$ is an irreflexive $pK_2 \cup qK_1$
\item
$H$ is an irreflexive $K_{2,2}$
\item
$H$ is a star in which the center has a loop and the other vertices do not
\item
$H$ is an irreflexive $pK_2$ together with a disjoint reflexive $qK_1$.
\end{enumerate}

Otherwise, the correspondence $H$-homomorphism problem is NP-complete.
\end{theorem}

For the correspondence $H$-list-homomorphism problem, the classification is
the same, except the cases (3) and (5) are NP-complete.

\begin{theorem}\label{listmixed}
Let $H$ be a graph with possible loops. Suppose moreover that if $H$ has both a 
vertex with a loop and a vertex without a loop, then it has no isolated loopless vertices.

Then the correspondence $H$-list homomorphism problem is polynomial-time 
solvable in cases (1, 2, 4, 6, 7), and is NP-complete otherwise.
\end{theorem}

Note that the graphs in cases (1-3) are reflexive, in cases (4-5) irreflexive, and
cases (6-7) mix loops and non-loops.

In the process of proving Theorem \ref{mixed}, we also classified the complexity 
of a bipartite version of the correspondence $H$-homomorphism problems and 
(list)-$H$-homomorphism problems. Specifically, we assume that $H$ is a bipartite
graph with a set of {\em black vertices} and a disjoint set of {\em white vertices}.
The {\em by-side correspondence $H$-homomorphism problem} asks whether an
input bipartite graph $G$ (with edge-labeling $\ell$) admits an $\ell$-correspondence
homomorphism to $H$ taking black vertices of $G$ to black vertices of $H$ and
white vertices of $G$ to white vertices of $H$. The {\em by-side correspondence 
$H$-list homomorphism problem} asks whether an input bipartite graph $G$ (with 
edge-labeling $\ell$ and lists $L(x), x \in V(G)$, such that for black vertices $x$ the
lists $L(x)$ contain only black vertices, and similarly for white vertices) admits an 
$\ell$-correspondence list homomorphism to $H$.

\begin{theorem}\label{bipa}
Let $H$ be a bipartite graph. Then the by-side correspondence $H$ homomorphism 
problem is polynomial-time solvable in case (4) above, as well as
\begin{enumerate}
\item[(8)]
$H$ is a complete bipartite graph plus any number of isolated vertices,
\item[(9)]
$H$ is a tree of diameter $3$ plus any number of isolated vertices,
\item[(10)]
$H$ consists of two disjoint copies of $K_{1,2}$ with white leaves, plus 
any number of black isolated vertices,
\item[(11)]
$H$ consists of two disjoint copies of $K_{2,2}$.
\end{enumerate}
In all other cases it is NP-complete.
\end{theorem}

For the list version we have the following result.

\begin{theorem}\label{listbipa}
Let $H$ be a bipartite graph. Then the by-side correspondence $H$-list homomorphism 
problem is polynomial-time solvable in cases (4, 8, 9) above, and is NP-complete otherwise.
\end{theorem}

In the Appendix, we sketch the proofs for Theorems \ref{mixed}, \ref{listmixed}, \ref{bipa}, and \ref{listbipa}. 

\section*{Acknowledgement}
The second author wishes to acknowledge the research support of NSERC Canada, through a Discovery Grant.

\hspace{2mm}

\newpage

\section{Appendix}

We sketch the proofs of all statements in Theorems \ref{mixed}, \ref{listmixed}, \ref{bipa}, 
and \ref{listbipa}, by case analysis.

{\bf Suppose $H$ is a bipartite graph.} 
If $H$ contains a connected component 
$C$ with $|C \cap U |\geq 3$ or $|C \cap V| \geq 3$, say the former. Then $H^2$ is 
reflexive and contains at least the two components on $C \cap U$ and $C \cap V$. 
By a reduction from Theorem \ref{main} the correspondence $H$-homomorphism 
problem is NP-complete.
If $H$ contains at least two components with at least one edge each, then $H^2$
has at least four components which must all be loops or isolated vertices, else the 
correspondence $H$-homomorphism problem is again NP-complete by Theorem 
\ref{main}. This covers all $H$ except case (4), which is easily solved by a trivial 
algorithm (similar to the case of a reflexive co-clique in Theorem \ref{main}), even 
with lists.

If $H$ contains only one component $C$ other than an isolated vertex, then the 
remaining cases are (a) $C$ is a path $uvu'$, (b) $C$ is a path $uvu'v'$, and (c) 
$C$ is $K_{2,2}$. In case (a) $H^2$ contains looped components of size one and 
two, which has NP-complete correspondence $H$-homomorphism problem by 
Theorem \ref{main}. In case (b) we introduce a loop with the permutations $(uvu'v')$ 
and identity, which eliminate only $v'$, since $v'u$ is not an edge of $H$, leaving us 
with the path $uvu'$ which has NP-complete correspondence $H$-homomorphism 
problem by case (a). (See Proposition \ref{last} for a similar technique of loop addition.) 
In case (c) $H^2$ has two reflexive components of size two plus isolated vertices. 
An isolated vertex can simulate lists in $C$, so either with lists or with at least one 
isolated vertex, case (c) is NP-complete by Theorem \ref{main}.

Finally, the case (5) $H=K_{2,2}$ the correspondence $H$-homomorphism problem
is polynomial again by Gaussian elimination. If we view $H$ as $K_{2,2}$ with parts 
$\{00,01\},\{10,11\})$, then we have the adjacency $(xy)(zt)$ ifand only if $z=x+1$ 
modulo $2$, and for a permutation $\rho$  we have $\rho(xy)=zt$ if and only if 
$zt=\rho(00)+(\rho(10)-\rho(00))(xx)+(\rho(01)-\rho(00))(yy)$ modulo $2$.
Thus the problem again reduces to linear equations modulo $2$.

{\bf We now consider the by-side problems}.
Suppose $H$ contains an induced path $uvu'v'u''$.
We show that the by-side correspondence $H$-homomorphism problem
is NP-complete. If $uu''$ is not an edge
in $H^2$, then the corresponding component $H^2\cap (U\times U)$ is reflexive
but not complete, so the problem is NP-complete by Theorem \ref{main};
otherwise there is a path $u''v''u$, giving a cycle $uvu'v'u''v''$ plus
possibly the edge $u'v''$. Consider the permutation $\rho=(uu')$ and join
two variables $x, y$ by an edge labelled by two identity permutations, and 
another parallel edge labelled by $(\rho,\rho)$.
If the edge $u'v''$ was present, then the edge $u'v'$ will be removed,
and we still have the induced path $v'u''w''uv$, and we can repeat the
process. Otherwise $u'v''$ is absent. If there is a vertex $w\neq u'$
and a path $vwv'$, there will again remain the induced path $uvwv'u''$, and
we can repeat. Otherwise no vertex is adjacent to the two vertices
of the $6$-cycle. If there is an edge $vw$, with $w$ not in the $6$-cycle,
then the edge $wu''$ is in $H^2$, so there is a path $vww'u''$, and
there will remain the induced path $vww'u''v'$, again we can repeat. 
Finally if the component is just the $6$-cycle, there will remain
at least components $uvu'$, $v'u''v''$, which for $H^2\cap (U\times U)$
gives components of $uu'$,$u''$ for which the correspondence 
$H$-homomorphism problem is NP-complete from Theorem \ref{main}.

In the remaining case, for each component $C$,
if $u, u'\in U\cap C$ and $v, v'\in V\cap C$, then one of $u, u'$
dominates the other (similarly one of $v, v'$ dominates the other).
Otherwise there exists an induced subgraph $uv, u'v'$ and a path $uwu'$,
giving an induced path $vuwu'v'$.
Order the vertices $u_1,u_2,\ldots, u_k$ and $v_1,v_2,\ldots, v_{\ell}$
so that $u_i$ dominates $u_{i+1}$, and $v_j$ dominates $v_{j+1}$.
Let $a_i,b_j$ be the degrees of $u_i,v_j$ respectively.
If $a_2 > a_k$, the by-side correspondence $H$-homomorphism problem is NP-complete. 
If $H$ has more that one component, then $H^2$ has a component of size three or 
more, so NP-completeness follows as in the reflexive case. We thus assume $H$ 
is connected. We may assume $a_1=a_2$ and $a_3=a_k$. Indeed, if we let $\pi$ 
be the identity and $\rho=(u_i u_{i'})$ with $i<i'$ and in parallel another edge with 
both $\pi \rho$ the identity, this reduces the neighbors of $u_i$ to those of $u_{i'}$. 
Note that $b_1>b_{\ell}$, and we may similarly assume $b_2=b_{\ell}$.
Then $xy\in E(H)$ if and only if $x\in\{a_1,a_2\}\vee y=b_i$. By permutation, 
we also get $x\in\{a_1,a_3\}\vee y'=b_{i'}$ and $x\in\{a_2,a_3\}\vee y'=b_{i''}$.
This gives $y=b_i\vee y'=b_{i'}\vee y''=b_{i''}$ with $i,i',i''\in\{1,2\}$. Thus there is 
a natural reduction from $3$-satisfiability, and the by-side correspondence 
$H$-homomorphism problem is NP-complete. 

Assume $H$ is connected. Then the remaining cases are
(3) $a_1=a_k$ and $b_1=b_{\ell}$, and  (4) $a_1>a_2=a_k$ and $b_1>b_2=b_{\ell}$.
In both these cases the by-side correspondence $H$-list-homomorphism problem 
is polynomial. For (3), any choice from the lists solves the problem.
For (4), $xy\in E(H)$ iff $x=a_1\vee y=b_1$.
By permutation we get boolean clauses $x_i\vee y_j$
where $x_i$ stands for $x=a_i$
and $y_j$ stands for $y=b_j$. We also add clauses
$\overline{x_i}\vee\overline{x_{i'}}$ for $i\neq i'$
$\overline{y_j}\vee\overline{y_{j'}}$ for $j\neq j'$.
Solve this instance of $2$-satisfiability. This gives a partial assignment
$x=a_i$, $y=b_i$ to variables $x, y$; the remaining variables may be
assigned arbitrarily.

If $H$ is not connected, in case (4) the by-side problem follows from the bipartite
version above, which allows lists. Cases (8),(9) are still polynomial with lists; adding 
an isolated vertex to either or both $U, V$ just simulates lists on the appropriate side.
By Theorem \ref{main} applied to $H^2$, these three are the only polynomial 
by-side cases with lists or or those without lists but with only one component
having more than one vertex.

The remaining by-side cases have $H$ having two components plus isolated
vertices, where the components are (*) either $K_{2,2}$, or a path $P_4$ on
four vertices, with no isolated vertices; and (10) two $K_{1,2}$ (with white leaves)
plus isolated vertices on the black side. Any $P_4$ can have one side permuted 
in parallel with the $P_4$, giving $K_{1,2}$ plus one isolated vertex, making this
case of (*) NP-complete. Suppose $H$ is the disjoint union of $H_1$, and $H_2$ 
obtained from $H_1$ by flipping the two sides as a by-side correspondence 
$H$-homomorphism problem. Then $H_1$ can be viewed as a standard 
(not by-side) correspondence $H$-homomorphism problem, which if polynomial 
makes the by-side correspondence $H$-homomorphism problem polynomial. 
Thus the remaining case (11) of two $K_{2,2}$, is polynomial for the by-side 
correspondence $H$-homomorphism problem. Finally (10) is polynomial with 
one boolean variable for the black vertices and two boolean variables for the 
white vertices.

{\bf For graphs $H$ that have both loops and non-loops,} 
we may assume $H$ has no isolated loopless 
vertices. We create a bipartite graph $H'=(U,V,E)$ 
that has for each edge $ww'$ in $H$, two edges $uv'$ and $vu'$. If the 
by-side problem (with or without lists) for $H'$ is NP-complete, then so is 
the corresponding problem for $H$. The only cases where both loops and 
non-loops are present are (6, 7), both easily checked to be polynomial even with lists.

\end{document}